\documentclass{article}
\usepackage[ansinew]{inputenc}
\usepackage[spanish,english]{babel}
\usepackage{amsmath,amsfonts,amssymb,amsthm}
\usepackage{graphicx}
\usepackage[usenames,dvipsnames]{color}
\usepackage{color}
\usepackage{lineno}
\usepackage{enumerate}
\usepackage{graphicx}
\usepackage[colorlinks=true,citecolor=black,linkcolor=black,urlcolor=blue]{hyperref}
\usepackage{epsfig, subfigure}
\usepackage{amscd,latexsym}
\usepackage{float}

\theoremstyle{plain}
\newtheorem{theorem}{Theorem}
\newtheorem{lemma}[theorem]{Lemma}

\newtheorem{prop}[theorem]{Proposition}

\newtheorem{observation}[theorem]{Observation}

\newcommand\blfootnote[1]{%
	\begingroup
	\renewcommand\thefootnote{}\footnote{#1}%
	\addtocounter{footnote}{-1}%
	\endgroup
}

\title{On circles enclosing many points}

\author{Merc\`e Claverol\thanks{{Email: merce.claverol@upc.edu}. Research supported by project MTM2015-63791-R (MINECO/FEDER), and by project Gen. Cat. DGR 2017SGR1640. Universitat Polit\`ecnica de Catalunya, Spain}\and
Clemens Huemer\thanks{{Email: clemens.huemer@upc.edu}. Research supported by project MTM2015-63791-R (MINECO/FEDER), and by project Gen. Cat. DGR 2017SGR1336. Universitat Polit\`ecnica de Catalunya, Spain}
\and Alejandra  Mart\'inez-Moraian\thanks{{Email: alejandra.martinezm@uah.es.}Universidad de Alcal\'a, Spain}}

\begin{document}
\maketitle

\blfootnote{\begin{minipage}[l]{0.3\textwidth} \includegraphics[trim=10cm 6cm 10cm 5cm,clip,scale=0.15]{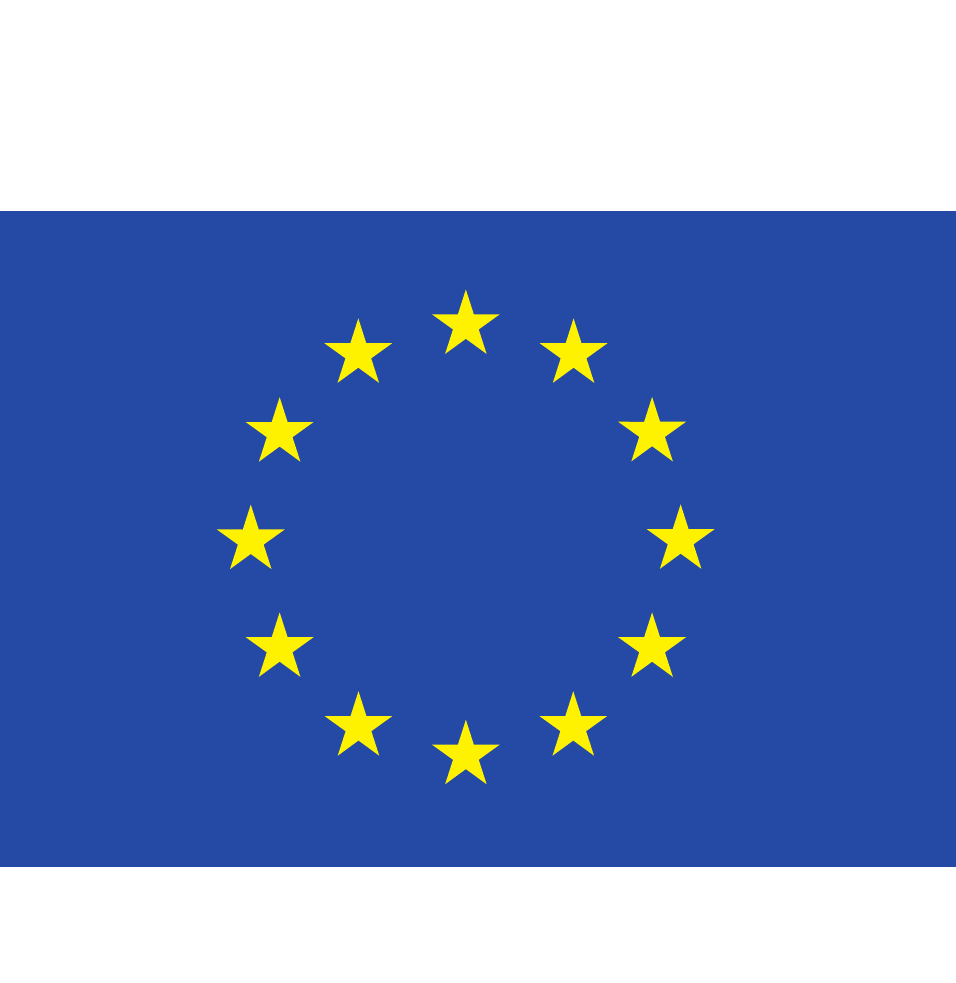} \end{minipage}  \hspace{-2cm} \begin{minipage}[l][1cm]{0.7\textwidth}
		This project has received funding from the European Union's Horizon 2020 research and innovation programme under the Marie Sk\l{}odowska-Curie grant agreement No 734922.
	\end{minipage}}

\begin{abstract}
We prove that every set of $n$ red and $n$ blue points in the plane contains a red and a blue point such that every circle through them encloses at least $n(1-\frac{1}{\sqrt{2}}) -o(n)$ points of the set. This is a two-colored version of a problem posed by Neumann-Lara and Urrutia. We also show that every set $S$ of $n$ points contains two points such that every circle passing through them encloses at most $\lfloor{\frac{2n-3}{3}}\rfloor$ points of $S$. The proofs make use of properties of higher order Voronoi diagrams, in the spirit of the work of Edelsbrunner, Hasan, Seidel and Shen on this topic. Closely related, we also study the number of collinear edges in higher order Voronoi diagrams and present several constructions.
	
%Insert your abstract here. Include keywords, PACS and mathematical
%subject classification numbers as needed.
%\keywords{Discrete geometry \and Point set \and Circle containment \and Voronoi diagram}
\end{abstract}

\section{Introduction}
\label{intro}

Neumann-Lara and Urrutia~\cite{NU88} posed the following problem: Prove that every set $S$ of $n$ points in the plane contains two points $p$ and $q$ such that any circle which passes through $p$ and $q$ encloses ``many" points of $S$. The question is to quantify this number of enclosed points which can always be guaranteed. We only consider point sets without three collinear points and without four cocircular points. We say that such a point set is in general position. Almost all the results on this question date back to the late 1980's. The first bound $\lceil \frac{n-2}{60} \rceil$, by Neumann-Lara and Urrutia~\cite{NU88}, was improved  in a series of papers~\cite{BSSU80,H89,HRW89}. The best bound was obtained by Edelsbrunner et al.~\cite{EHSS89} who proved that any set of $n$ points in the plane contains two points such that any circle through those two points encloses at least $n(\frac{1}{2}- \frac{1}{\sqrt{12}})+O(1) \approx \frac{n}{4.7}$ points. Their proof makes use of properties of higher order Voronoi diagrams.  20 years later, Ramos and Via\~na~\cite{RV09} made progress and proved a stronger statement: There is always a pair of points such that any circle through them has, both inside and outside, at least $\frac{n}{4.7}$ points. 
To prove their result, they transformed the problem from circles in the plane to planes in $\mathbb{R}^3$ and used results on the number of $j$-facets of point sets in $\mathbb{R}^3$.\\ 
The known upper bound on the number of enclosed points is $\lceil\frac{n}{4}\rceil-1$, due to a construction by Hayward et al.~\cite{HRW89}. Urrutia~\cite{Urrutia} conjectured that $\frac{n}{4}\pm c$, for some constant $c$, is the tight bound.  For sets of $n$ points in convex position a tight bound of $\lceil\frac{n}{3}\rceil-1$ is known~\cite{HRW89}.\\
Our contributions to this problem are the following ones:\\
$\bullet$ In Section~\ref{sec:adapt} we present a modified and shorter version of the proof of Edelsbrunner et al.~\cite{EHSS89} which also leads to the result of  Ramos and Via\~na~\cite{RV09}. As in~\cite{EHSS89}, the proof makes use of properties of higher order Voronoi diagrams.\\ 
$\bullet$ The proposed modification gives rise to the following result, shown in Section~\ref{sec:ramsey}. 
Every set of $n$ points in the plane contains two points such that any circle through those two points encloses at most  $\lfloor{\frac{2n-3}{3}}\rfloor$ points of $S$.\\
$\bullet$ In Section~\ref{sec:twocolored} we %study a two-colored version of the problem and 
%: Prove that every set $S$ of $n$ red points and $m$ blue points in the plane contains a red point $p$ and a blue point $q$ such that any circle which passes through $p$ and $q$ encloses ``many" points of $S$. 
show that sets $S$ of $n$ red points and $m$ blue points contain a red point $p$ and a blue point $q$ such that any circle passing through them encloses at least $\frac{n + m - \sqrt{n^2 + m^2}}{2} - o(n+m)$ points of $S$. For $n=m$ this gives the bound $n(1-\frac{1}{\sqrt{2}}) -o(n)  \approx 0.2928 n$. 
The colored version of the problem was studied by Prodromou~\cite{P07} for any dimension $d$ and $\lfloor\frac{d+3}{2}\rfloor$ colors. The particular case $d=2$ in~\cite{P07}, Theorem 1.1, gives a lower bound of $\frac{n+m}{36}$ on the number of enclosed points. Our result improves this bound. We also present an upper bound construction with $n$ red and $n$ blue points in convex position.\\ 
%We also present a construction $S$ of $n$ red and $n$ blue points in convex position where for every pair of points, one red and the other blue, there is a circle through them which encloses at most $\lfloor\frac{n}{2}\rfloor$ points of $S$.\\%states that any set $S$ of $n$ red and $m$ blue points in the plane contains a red point $p$ and a blue point $q$ such that any circle through $p$ and $q$ encloses at least $\frac{n+m}{36}$ points of $S$. Our result improves this bound.\\  
$\bullet$ In Section~\ref{sec:repeat} we study how many circles passing through two given points $p$ and $q$ (and a third point of $S$) enclose the same number of points of $S$. This is equivalent to study how many  edges of the order-k Voronoi diagram of $S$ can lie on the same line.
Apart from~\cite{KS94} not much seems to be known on this question. 
%It is known that the maximum number of (bounded) edges on the bisector line of two sites in the order-k Voronoi diagram and in the order-(n-k) Voronoi diagram of $S$ is at most $\min\{k,n-k\}$~\cite{KS94}. 
We present some constructions with many collinear edges in higher order Voronoi diagrams. 
We believe that this line of research can also lead to an improvement for the problem posed by Neumann-Lara and Urrutia.
See~\cite{AIUU96,SSW08} for related works.

\section{An adaption of the proof by Edelsbrunner, Hasan, Seidel and Shen}\label{sec:adapt} 

%We get another proof of the bound of Edelsbrunner et al.~\cite{EHSS89}, in the stronger version of Ramos et al: every circle through the two points leaves at least $n/4.7$ points inside and outside~\cite{RV09}. Check if it is correct. The proof starts as the proof of Edelsbrunner, but then almost no calculations are needed, when using a formula for the number of points in circles and a bound for k-sets. Task: extend the proof (if correct) to further improve the bound. Also, it looks very feasible to adapt it to the colored version.\\

%The proof idea of~\cite{EHSS89} is as follows. Let $S$ be a set of $n$ points in general position in the plane. For each pair of points $p$, $q$ of $S$, let $b_{pq}$ be the perpendicular bisector of the segment $\overline{pq}$. 
%Consider this set $B$ of ${{n}\choose{2}}$ lines. Each such line $b_{pq} \in B$ is subdivided into $n-1$ segments, where each such segment is an edge of some order-k %Voronoi diagram. Edelsbrunner et al. show that the sum of all the edges of the first order-k Voronoi diagrams, for $k=1, 2,\ldots, \approx (1/2-1/\sqrt{12})n$, is less than ${{n}\choose{2}}$. This implies that one of the lines of $B$ does not use any edge of the first order-$k$ Voronoi diagrams among its $n-1$ segments. Consequently, for the two points $p$ and $q$ which define this line $b_{pq}$, any circle through them encloses at least $n(1/2-1/\sqrt{12})$ many points of $S$. We explain this in some more detail.\\

Let $S$ be a set of $n$ points in general position. For each pair of points $p$, $q$ of $S$, let $b_{pq}$ be the perpendicular bisector of the segment $\overline{pq}$. Let $B$ be this set of ${{n}\choose{2}}$ perpendicular bisectors. For each pair of points $p, q$ of $S$, define $C_{pq}$ as the set of circles passing through them. The center of each circle in $C_{pq}$ is on $b_{pq}$. Any line $b_{pq}$ in $B$ is cut into $n-1$ open segments (two of them are unbounded) which are delimited by the center points of the circles passing through $p$, $q$, and one of the $n-2$ remaining points of $S.$ Any circle in $C_{pq}$ with center on one of these segments encloses the same subset of points of $S$; its cardinality is the {\it{weight}} of the segment. It is well known that if the weight of such a segment is $k-1$, then the segment is an edge of the order-k Voronoi diagram, see e.g.~\cite{L82}.\\ 
The proof idea of Edelsbrunner et al.~\cite{EHSS89} is as follows. They show that the sum of all the edges of the first order-k Voronoi diagrams, for $k=1, 2,\ldots, \approx (1/2-1/\sqrt{12})n$, is less than ${{n}\choose{2}}$. This implies that one of the lines in $B$ does not use any edge of the first order-$k$ Voronoi diagrams among its $n-1$ segments. Consequently, for the two points $p$ and $q$ which define $b_{pq}$, any circle through them encloses at least $n(1/2-1/\sqrt{12})$ points of $S$. We present a similar proof. %The next observation, applied in~\cite{EHSS89} is essential for the proof. 
\begin{observation}\label{obs:1}
	When moving the center of a circle in $C_{pq}$ along $b_{pq}$ from one segment to the consecutive one, the number of points contained in the two corresponding circles differs by $\pm 1$. Equivalently: The weights of two consecutive segments on $b_{pq}$ differ by $\pm 1$.
\end{observation}
%Also note that if an unbounded (the first or the last) segment of a line $\ell_{p,q}$ corresponds to a circle which encloses $k$ points, then the segment $\overline{pq}$is a {\it{$k$-edge}}. In that case, we say that $\ell_{p,q}$ corresponds to a $k$-edge, the supporting line of $p$ and $q$, and the first segment is a $k$-segment.
For $0 \leq j \leq \frac{n-2}{2}$, a segment $\overline{pq}$  connecting two points $p$, $q$ of $S$ is a {\it{$j$-edge}} of $S$, if the line through $p$ and $q$ divides the plane into two open half-planes, such that one of them contains $j$ points of $S$.
\begin{observation}\label{obs:2}
	Let $\overline{pq}$ be a $j$-edge, and let $b_{pq}$ be the perpendicular bisector of  $\overline{pq}$. Then the two unbounded segments of $b_{pq}$ have weights $j$ and $n-j-2$. 
\end{observation}

%This follows easily, just observe that when moving the center of a circle of $C_{pq}$ along $b_{pq}$ towards infinity, the circle approaches a half-plane defined by the supporting line of $\overline{pq}$, which separates $S$ into two subsets of $j$ and $n-j-2$ points.

%Observations~\ref{obs:1} and~\ref{obs:2} imply:
\begin{observation}\label{obs:3}
	Let $\overline{pq}$ be a $j$-edge of $S$, and let $b_{pq}$ be the perpendicular bisector of  $\overline{pq}$. For every $j\leq k \leq n-j-2$, the line $b_{pq}$ contains at least one segment of weight $k$.
\end{observation}

Let $c_k$ be the number of circles passing through three points of $S$, that enclose exactly $k$ points of $S$.
Differing from the proof of~\cite{EHSS89}, we will use the following property, obtained by Lee~\cite{L82}; later proofs were given in~\cite{A04, CS89, L03}:
\begin{equation}\label{equ:circles}
c_k+c_{n-k-3}=2(k+1)(n-k-2)
\end{equation}
A direct correspondence between the numbers $c_k$ and the number of faces of higher order Voronoi diagrams is given for instance in~\cite{L03}.
\begin{observation}\label{obs:circles}
	A circle passing through points $a,b,c \in S$ corresponds to three segments, one on $b_{ab}$, one on $b_{ac}$ and one on $b_{bc}.$
	Hence, when summing over all ${{n}\choose{2}}$ lines in $B$, the number of segments of weight $k$ plus the number of segments of weight $n-k-3$ equals $3c_k+3c_{n-k-3}=6(k+1)(n-k-2)$.
\end{observation}

We state now the adapted proof idea: Using Observation~\ref{obs:circles} we show that at least one of the lines in $B$, say $b_{pq}$, neither contains a segment of weight $k-1$ nor of weight $n-k-2$, for some value $k\leq \frac{n-4}{2}$ to be determined.  Then, each of  the $n-1$ segments of $b_{pq}$ has weight between $k$ and $n-k-3$. Therefore, any circle passing through points $p$ and $q$, which define $b_{pq}$, encloses between $k$ and $n-k-3$ points of $S$. It turns out that $k=\left(\frac{1}{2}-\frac{1}{\sqrt{12}}\right)n$ is the best choice for $k$ in the proof.\\

{\bf{Claim}:} There exist two points $p, q \in S$ such that the  perpendicular bisector $b_{pq}$ of  $\overline{pq}$ neither contains a segment of weight $k-1$ nor of weight  $n-k-2$.\\

Assume by contradiction, that every line in $B$  contains a segment of weight $k-1$ or of weight $n-k-2$. Partition the lines in $B$ into three classes:  (1) those whose defining points $p,q \in S$ form a $j$-edge for $j \leq k-1$, (2) those whose defining points $p,q \in S$ form a $j$-edge for $j\geq k+1$, and (3) the lines corresponding to $k$-edges.

%For simplicity, we can ignore those lines for which we have $j=k$ or $j=k-1$ or $j=k+1$, because there are less than $O(n \sqrt[3]{k+1})$ of them (the number of $k$-edges is less than $O(n \sqrt[3]{k+1})$~\cite{D98}, and we will see in the following that this does not affect the asymptotic counting.\\

%Each line of type (1) contains at least one segment of weight $k-1$ and at least one of weight $n-k-1$, because along such a line we go from $j-2$ to $n-j-2$ in steps %of $1$.

By Observation~\ref{obs:3} each line of type (1) contains at least  one segment of weight $k$ and at least one of weight $n-k-3$.\\
For any line $b_{pq}$ of type (2), its unbounded segments have weights $j\geq k+1$ and $n-j-2 \leq n-k-3$, respectively. 
By assumption, $b_{pq}$ also contains a segment of weight $k-1$ or of weight $n-k-2$. If $b_{pq}$ contains a segment of weight $k-1$, then $b_{pq}$ contains two segments of weight $k$; indeed, when traversing $b_{pq}$, we go from a segment of weight $j\geq k+1$ via one of weight $k-1$ to one of weight $n-j-2 \geq k+1$; since the changes of the weights of consecutive segments are $\pm1$, we encounter a subsequence of weights $k+1,k,k-1,k,k+1$ among the segments of $b_{pq}$.
In the same way, if $b_{pq}$ contains a segment of weight $n-k-2$, we encounter a subsequence of weights $n-k-4, n-k-3, n-k-2, n-k-3$ when traversing $b_{pq}$. Hence, in this case $b_{pq}$ contains two segments of weight $n-k-3$.
We conclude that each line of type (1) and of type (2) in $B$ contains at least two segments with weight in the set $\{k, n-k-3\}$. 
%Since when transversing $b_{pq}$, the changes of the weights of segments are in steps of $1$, $b_{pq}$ either contains at least two segments of weight $k$ (in the %case $b_{pq}$ contains a segment of weight $k-1$, we go "down and up again", i.e. we see a subsequence of weights $k+1,k,k-1,k,k+1$) or $b_{pq}$ contains at least two% %segments of weight $n-k-3$ (in the case $b_{pq}$ contains a segment of weight $n-k-2$, we go "up and down again", i.e. we see a subsequence of weights $n-k-4, n-k-3, %n-k-2, n-k-3, n-k-4$ ).\\

The number of lines of type (3) is at most $O(n \sqrt[3]{k+1})$, the known upper bound on the number of $k$-edges~\cite{D98}. %This also holds for $k-1$-edges and $k+1$-edges.

%We obtain that each of the ${{n}\choose{2}} - O(n \sqrt{n})$ lines of type (1) and of type (2) of $B$ contains at least two segments with weight from the set $\{k, n-k-3\}$. 
%, excluding the at most $O(n\sqrt{k})$ lines corresponding to $k$-edges, $k-1$-edges, and $k+1$-edges, contains at least two segments with weight from the set $\{k, n-k-3\}$.
By Observation~\ref{obs:circles}, the number of segments of weight $k$ or of weight $n-k-3$ among all lines in $B$ is  $$6(k+1)(n-k-2).$$
We get a contradiction if $$6(k+1)(n-k-2) < 2\left({{n}\choose{2}}- O(n \sqrt[3]{k+1})\right),$$ because then there would not be enough segments to cover all the lines of type (1) and of type (2) in $B$ with two segments. %Note that we need two segments per line.
The largest value of $k$ which gives a contradiction is $k=\left(\frac{1}{2}-\frac{1}{\sqrt{12}}\right)n -o(n)$.
This proves the claim.\\

Therefore, there exist two points $p,q \in S$ such that $b_{pq}$ contains no  segment of weight $k-1$ and no segment of weight $n-k-2$. Note that this line $b_{pq}$ is of type (2). Then,  $b_{pq}$ cannot contain a segment of weight $i$ for $i<k-1$ and for $i>n-k-2$ either.  We thus have obtained another proof of the theorem by Ramos and Via\~na~\cite{RV09}, when neglecting sublinear terms:
\begin{theorem}\label{thm:general}
	Every set $S$ of $n$ points in general position in the plane contains two points such that each circle passing through them encloses at least $k$ and at most $n-k-3$ points of $S$, for $k=\left(\frac{1}{2}-\frac{1}{\sqrt{12}}\right)n -o(n)$.
\end{theorem}

\section{Circles enclosing not too many points}\label{sec:ramsey}

The following lemma is implied by known results on higher order Voronoi diagrams and bounds on $k$-sets. A $k$-set of a point set $S$ is a subset of $k$ points of $S$ which can be separated from the remaining points of $S$ by a straight line. 

\begin{lemma}\label{lem:unboundedregions}
Let $S$ be set of $n$ points and let $k< \frac{n-3}{2}$. Then $c_{k} \geq (k+1)(n-k-2)$ and $c_{n-k-3} \leq (k+1)(n-k-2)$. 
\end{lemma}

\begin{proof}
 Denote with $f_k^\infty$ the number of unbounded regions in the order-$k$ Voronoi diagram of $S$. Also define 
$f_0^\infty =0$ and $c_{-1}=0.$ These numbers $f_k^\infty$ are related to circles enclosing points via the following relation, see~\cite{L03}.
$$\sum_{i=1}^{k} f_{i-1}^\infty = (k-1)(2n-k) - c_{k-2}.$$
On the other hand, it is well known that each unbounded region $f_k^\infty$ corresponds to a $k$-set. The number of $\leq k$-sets of $S$ is known to be at most $k \cdot n$ for $k<\frac{n}{2}$~\cite{AG86}. Therefore, $\sum_{i=1}^{k} f_{i-1}^\infty \leq (k-1) \cdot n$ and 
$$c_{k-2} \geq (k-1)(2n-k) - (k-1) \cdot n.$$
Then,  $c_{k} \geq (k+1)(2n-k-2) - (k+1)\cdot n = (k+1)(n-k-2)$ and finally $c_{n-k-3} \leq (k+1)(n-k-2)$ follows from Property~(\ref{equ:circles}).
\end{proof}

A slight variation of the proof of Theorem~\ref{thm:general}, and using~\cite{AG86} instead of~\cite{D98}, leads to the following result.

\begin{theorem}\label{thm:ramsey}
	Let $S$ be a set of $n$ points in general position in the plane. Then $S$ contains two points such that 
	every circle passing through them encloses at most $\lfloor{\frac{2n-3}{3}}\rfloor$ points of $S$.
\end{theorem}

\begin{proof}
	The proof is by contradiction: Suppose that every line in $B$ contains a segment of weight $n-k-2$; we will see that the largest value of $k$ which gives a contradiction is $k= \lfloor{\frac{n-4}{3}}\rfloor$.\\
	Partition the lines in $B$ into two classes:  (1) those whose defining points $p,q \in S$ form a $j$-edge for $j\geq k+1$, and (2) those whose defining points $p,q \in S$ form a $j$-edge for $j \leq k$.\\ 
	%We first show that then the following holds for some $w \in \{k,n-k-3\}$: 

	{\bf{Claim}:} The total number of segments of weight $n-k-3$, among all lines in $B$, is at most $3(k+1)(n-k-2)$, and every line of type (1) contains two segments of weight $n-k-3$.\\ 
	
By Lemma~\ref{lem:unboundedregions}, $c_{n-k-3} \leq (k+1)(n-k-2)$ and by Observation~\ref{obs:circles}, the number of segments of weight $n-k-3$, among all lines in $B$, is at most $3(k+1)(n-k-2)$. By assumption, each line $b_{pq}$ in $B$ has a segment of weight $n-k-2$.  Let $b_{pq}$ be of type (1). By Observation~\ref{obs:2}, $b_{pq}$ has an unbounded segment of weight $n-j-2 \leq n-k-3$. When traversing $b_{pq}$ we go from an unbounded segment of weight $j\leq (n-2)/2$, via a segment of weight $n-k-3$ to a segment of weight $n-k-2$ and then to the other unbounded segment of weight at most $n-k-3.$ Hence, $b_{pq}$ contains two segments of weight $n-k-3$. This proves the claim.\\
	
	The number of lines in $B$ of type (2) is at most $(k+1) \cdot n$, because the number of $\leq k$-sets of $S$ is known to be at most $k \cdot n$~\cite{AG86} and the number of $k$-edges equals the number of $(k+1)$-sets~\cite{AG86,P85}.
	Since each line of type (2) contains a  segment of weight $n-k-3$ and each line of type (1) contains two segments of weight $n-k-3$, we get
	$$3(k+1)(n-k-2) \geq 1 (k+1)\cdot n + 2\left({{n}\choose{2}}-(k+1)\cdot n\right).$$ But this only holds if $k \geq   (n-3)/3$, and gives a contradiction for $k= \lfloor{\frac{n-4}{3}}\rfloor$. Therefore $B$ contains a line $b_{pq}$, all whose segments have weight at most  $n-k-3 = n-\lfloor{\frac{n-4}{3}}\rfloor-3= \lfloor\frac{2n-3}{3}\rfloor$. Then, every circle through $p$ and $q$ encloses at most  $\lfloor{\frac{2n-3}{3}}\rfloor$ points of $S$.
\end{proof}

\section{Two-colored point sets}\label{sec:twocolored}

%Let $R$ be a set of $n$ red points and $B$ be a set of $m$ blue points in the plane. We want to find a red point $p$ of $R$ and a blue point $q$ of $B$ such that any circle through $p$ and $q$ encloses many points of $R \cup B$.\\

\begin{theorem}
	Every set $S$ of $n$ red points and $m=\lfloor{c n\rfloor}$, for $c \in (0,1]$, blue points in general position in the plane contains a red point $p$ and a blue point $q$ such that any circle passing through them encloses at least $\frac{n + m - \sqrt{n^2 + m^2}}{2} - o(n+m)$ points of $S$. 
\end{theorem}
\begin{proof}
	The proof is very similar to the one of Theorem~\ref{thm:general}. The difference is that we now only consider bisectors $b_{pq}$ for points $p$ and $q$ of different color. Let $B$ be this set of $n m$ bichromatic bisectors. Each $b_{pq}$ in $B$ is cut into $n+m-1$ open segments, which are delimited by the center points of the circles passing through $p$, $q$, and one of the $n+m-2$ remaining points of $S.$ We then need a bound on the number of segments of weight $k$ plus the number of segments of weight $n+m-k-3$, among all bichromatic bisectors, analogous to Observation~\ref{obs:circles}. 
	\begin{observation}\label{obs:colorcircles}
		A circle passing through points $p_1, p_2, q \in S$, with $q$ of different color than $p_1$ and $p_2$, corresponds to two segments, one on $b_{p_1q}$  and one on $b_{p_2q}$.
		Hence, when summing over all $n m$ bichromatic bisectors in $B$, the number of segments of weight $k$ plus the number of segments of weight $n+m-k-3$ is at most $2c_k+2c_{n+m-k-3}=4(k+1)(n+m-k-2)$.
	\end{observation}
	Note that here we used Equation~(\ref{equ:circles}) which also counts circles passing through three points of the same color.
	Then, following the steps of the proof of Theorem~\ref{thm:general}, 
	we get that each line of type (1) and of type (2) in B contains at least
	two segments with weight from the set $\{k,n+m-k-3\}$; and the number of lines of type (3) is at most $O((n+m) \sqrt[3]{k+1})$.
	Then we get a contradiction if $$4(k+1)(n+m-k-2) < 2\left(n m- O((n+m) \sqrt[3]{k+1})\right),$$ because then there would not be enough segments to cover all the lines of type (1) and of type (2) in $B$ with two segments. %Note that we need two segments per line.
	Then, we can set $k = \frac{n + m - \sqrt{n^2 + m^2}}{2} - o(n+m)$. 
\end{proof}

Figure~\ref{fig:ejemplo} shows an upper bound construction, for $n$ red and $n$ blue points in convex position; for every pair of points, one red and the other blue, there is a circle through them which encloses at most $\lfloor{\frac{n}{2}}\rfloor$ points of $S$.

\begin{figure}[t]
	\begin{center}	
		\includegraphics[width=0.28\textwidth]{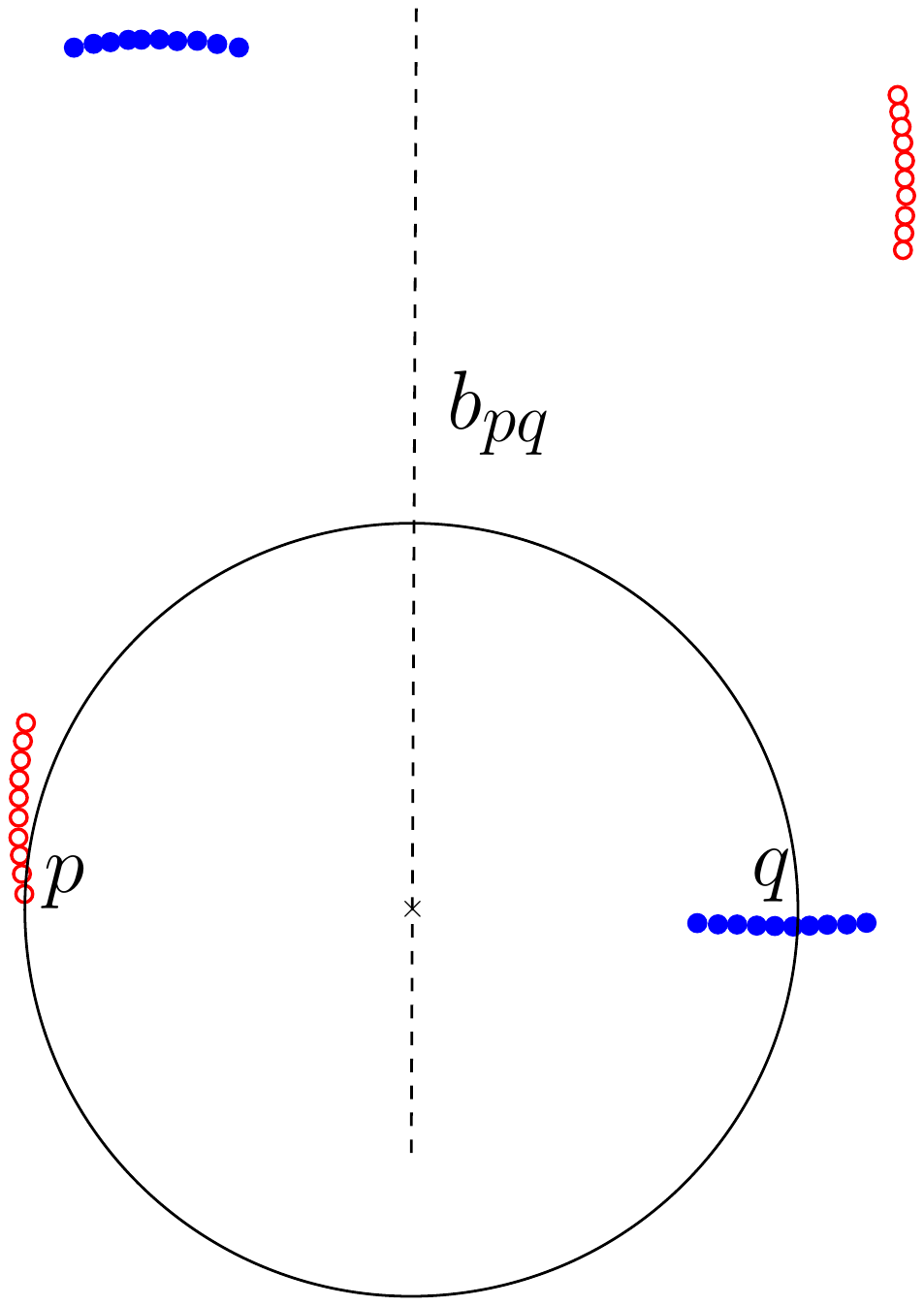}
		\caption{Convex configuration $S$ with $n$ red points and $n$ blue points, consisting of four groups of $\lfloor{\frac{n}{2}}\rfloor$ or $\lceil{\frac{n}{2}}\rceil$ points. 
			%Each group is placed very close to a corner of the unit square. 
		}
		\label{fig:ejemplo}
	\end{center}
\end{figure}

\section{Many segments of repeated weights}\label{sec:repeat}

%we study how many circles passing through two given points $p$ and $q$ (and a third point of $S$) enclose the same number of points of $S$. 
%We present point sets $S$ which contain two points $p$ and $q$ such that many circles through $p$ and $q$ (and another point of$S$) enclose the same number of points %of $S$. Equivalently, 

%We present sets of points for which we find many segments with repeated weights along a bisector of $B$. 

\begin{prop}\label{prop:recursive}
	There exists a set $S$ of $n$ points in general position in the plane which satisfies: %such that for at least $4n/7-O(1)$ different values of $k$, the order-k Voronoi diagram of $S$ satisfies: $\bullet$ There exist three bisectors of two sites that contributes with four edges to this diagram.
	Let $b_k$ be the number of bisectors among pairs of points of $S$ that contribute with four edges to the order-k Voronoi diagram of $S$. Then $\sum_{k=1}^{2n/7} b_k \geq \frac{4n}{7}-o(n)$ and $b_k \neq 0$ for at least $\frac{n}{6}-o(n)$ values of $k$. Further, only a subset of $O(\log(n))$ points is needed to define these bisectors. 
\end{prop} 

\begin{proof}
	The set of points $S$ is obtained recursively in the following way.
	Let $p$, $q$, $r$ be the vertices of an equilateral triangle and consider the supporting lines through them dividing the plane into seven regions, see Figure~\ref{fig:recursive}.
	In each of these regions there is a group of $n/7$ points such that the circle by $p$, $q$ and $r$ only contains points of the central region. We assign the points $p,q$ and $r$ to the central region. This is the initial configuration.
	
	\begin{figure}[hbt]
		\begin{center}	
			\includegraphics[width=0.65\textwidth]{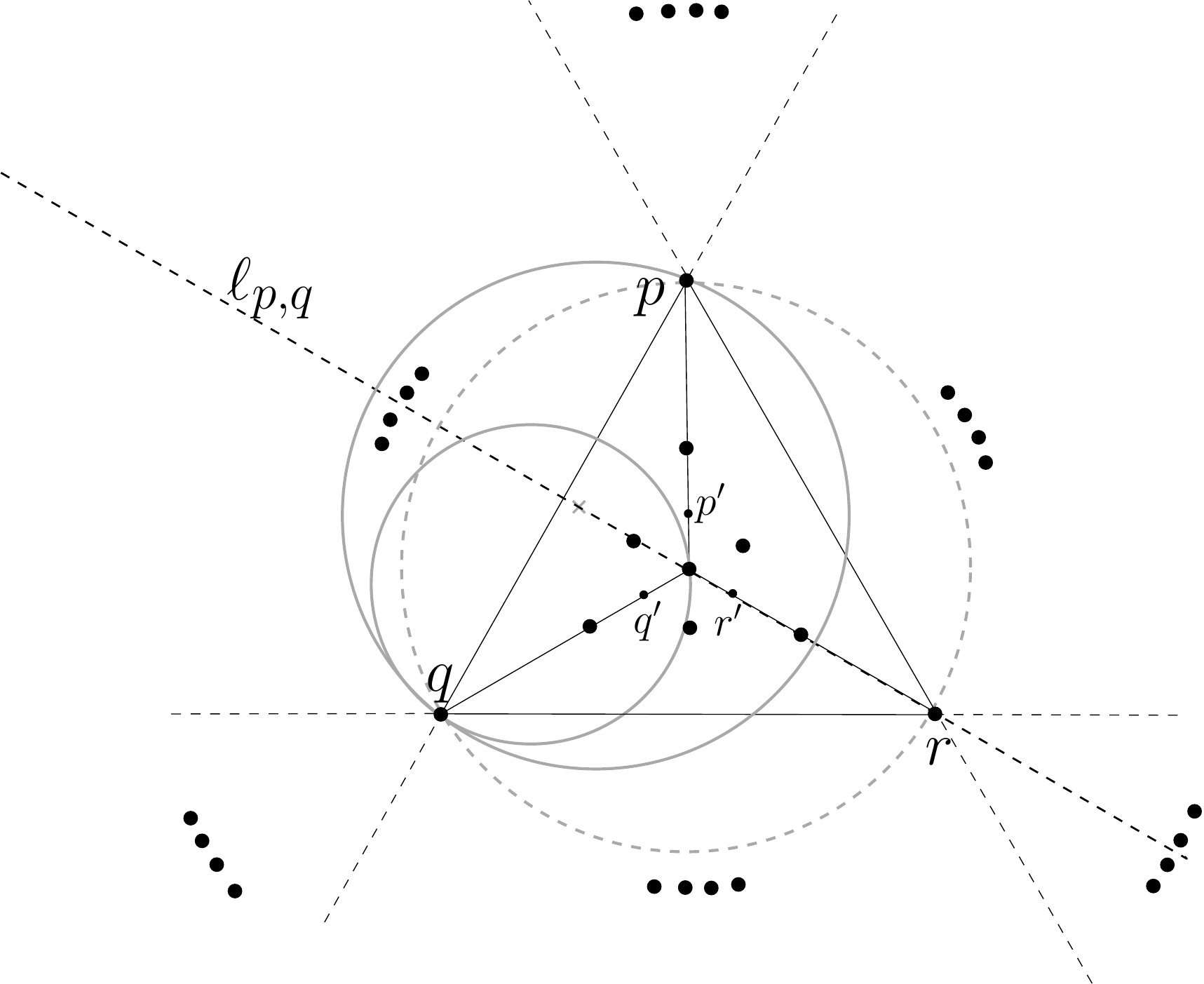}
			\caption{A recursive construction with many repeated weights.}
			\label{fig:recursive}
		\end{center}
	\end{figure}
	
	The $(n/7)-3$ points from the interior of the triangle $pqr$  give rise to seven other groups of $(n-3\cdot 7)/7^2$ points, arranged all of them in the same way as the initial configuration, that is, its central region has again three points $p'$, $q'$ and $r'$, which delimit another central region and so on. Note that  $p'$, $q',$ and $r'$ are placed very close to the center of the triangle $pqr$. Further, we can arrange the points symmetrically with respect to the lines through $\overline{pp'}$, $\overline{qq'}$ and $\overline{rr'}$; then we can move the points slightly to guarantee general position.
	
	First, let's see how many weights are repeated in any of the bisectors $b_{pq}$, $b_{pr}$ or $b_{qr}$.
	If we go through any of them, we obtain the list of weights: $[3n/7,\stackrel{\downarrow}{\ldots}, n/7, \stackrel{\uparrow}{\ldots} (2n/7)-3,\stackrel{\downarrow}{\ldots}, (n/7)-3, \stackrel{\uparrow}{\ldots}, (4n/7)-2]$, where consecutive values differ by one. Therefore, all values between $(2n/7)-3$ and $n/7$ are repeated four times. This already shows that there exists a bisector which contributes with four edges to $(n/7)-3$ different higher order Voronoi diagrams.
	Note that there is an interval of $b_ {pq}$ in which the circles through $p$ and $q$, and center in this interval, do not contain points outside of the central region defined by $p,q$ and $r$. 
	The same happens in the different central regions obtained recursively.
	Hence, considering this interval of the bisector $b_{p'q'}$ (analogously for $b_{p'r'}$,  $b_{q'r'}$), the list of weights is: $[\cdots, 3((n/7)-3)/7,\stackrel{\downarrow}{\ldots}, ((n/7)-3)/7, \stackrel{\uparrow}{\ldots} (2((n/7)-3)/7)-3,\stackrel{\downarrow}{\ldots}, ((n/7)-3)/7, \stackrel{\uparrow}{\ldots}, 4((n/7)-3)/7,\cdots]$, where $(((n/7)-3)/7)-3$ of them are repeated four times. To simplify the calculations, we can count approximately $(n/7^2)-3$ weights repeated four times, for each of  $b_{p'q'}$, $b_{p'r'}$, and  $b_{q'r'}$. Recursively, we obtain:
	$3(n/7+n/7^2+\cdots +n/7^{k-1}-3(k-1))$
	%$3((n/7)+n/7^2+\cdots +n/7^{k-1}-3(k-1)-(1/2)(1-1/7^{k-1})$
	%$=(1/4)(-2n+1)7^{1-k}+21\cdot 7^{-k}\,(k/2)+(n/2)-7/4$
	weights repeated four times.
	%O(n/2)
	
	Second, we consider any of the three bisectors of the segments $\overline{pp'}$, $\overline{qq'}$, $\overline{rr'}$. Their list of weights is $[n/2,\stackrel{\downarrow}{\ldots}, 2((n/7)-3))/7, \stackrel{\uparrow}{\ldots} 3((n/7)-3))/7,\stackrel{\downarrow}{\ldots},2((n/7)-3))/7, \stackrel{\uparrow}{\ldots}, n/2]$, where $((n/7)-3)/7$ (we use approximately $n/7^2$) values are repeated four times. 
	Considering the corresponding three halving lines in each step of the recursion we obtain: $3(n/7^2+\cdots +n/7^{k-1})$.
	%$=-(1/2)7^{1-k}n+21\cdot 7^{-k}(k/2)+(1/4)7^{1-k}+(n/14)-13/28.$
	%(n-7^2)/14
	
	Adding up,
	we obtain $\sum_{k=1}^{2n/7} b_k \geq 3(n/7+2n/7^2+\cdots +2n/7^{k-1}-3(k-1))=4n/7-(7n/7^k)-9$, which is approximately $4n/7$. Note that we can take $k\in O(\log(n))$. From the construction we get $b_k \neq 0$ for at least $n/7+n/7^2+\cdots +n/7^{k-1} - o(n) = n/6 -o(n)$ 
	values of $k$. 
\end{proof}

For sets of $2n$ cocircular points, the segments on the bisector of any $(n-1)$-edge (a halving line) have weight $n-1$. Another, not so elementary, construction without four co-circular points is given in Proposition~\ref{prop:halvinggeneral}, see Figure~\ref{fig:halvings3}.

%\begin{prop}
%	There exists a set $S$ of 2$n$ points in the plane such that every pair of points $p$, $q$ of $S$ which defines a halving line satisfies: Every circle passing through %$p$ and $q$ encloses $n-2$ or $n-4$ points. Equivalently: Every segment of $b_{pq}$ has weight $n-2$.  
%\end{prop}

%\begin{figure}[hbt]\label{fig:halvings}
%	\begin{center}	
%		\includegraphics[width=0.25\textwidth]{halvings.pdf}
%		\caption{All segments on any bisector of an $\frac{n-2}{2}$-edge have the same weight $\frac{n-2}{2}$.}
%		\label{fig:halvings}
%	\end{center}
%\end{figure}

\begin{prop}\label{prop:halvinggeneral}
	There exists a set $S$ of $2n$ points in general position in the plane such that every pair of points $p$, $q$ of $S$ which defines a halving line satisfies: Every circle passing through $p$ and $q$ encloses $n-2, n-1$ or $n$ points of $S$. %Equivalently: Every other segment of $b_{pq}$ has weight $n-2$.  
\end{prop}

\begin{proof}
	The construction is as follows. First, we place two points, $p_1$ and $q_1$; let $m$ be their midpoint. We consider the lines $\ell_1,\cdots, \ell_{n}$, where $\ell_1$ is the supporting line of $\overline{p_1q_1}$, $\ell_{n}$ is $b_{p_1q_1}$, and the remaining ones are obtained from $\ell_1$ after successive rotations of angle $\pi/(2n)$ and center $m$. In the following we define the points $q_i$, $p_i$ for $i=n,\cdots,3,2 $ such that all $p_i$ are above $\ell_1$ and all $q_i$ are below $\ell_1$. Both $q_i$ and $p_i$ will lie on $\ell_i$. 
	Now we place the points $p_{n}$ and $q_{n}$ on $b_{p_1q_1}$, in such a way $p_{n}$ is close to $m$, and $q_{n}$ cocircular with $p_1$, $q_1$ and $p_n$.
	%We denote $C_{i,i+1}$ the circle through $q_i,p_i,q_{i+1}$ and $p_{i+1}$.
	
		\begin{figure}[t!]
			\begin{center}	
				\includegraphics[width=0.65\textwidth]{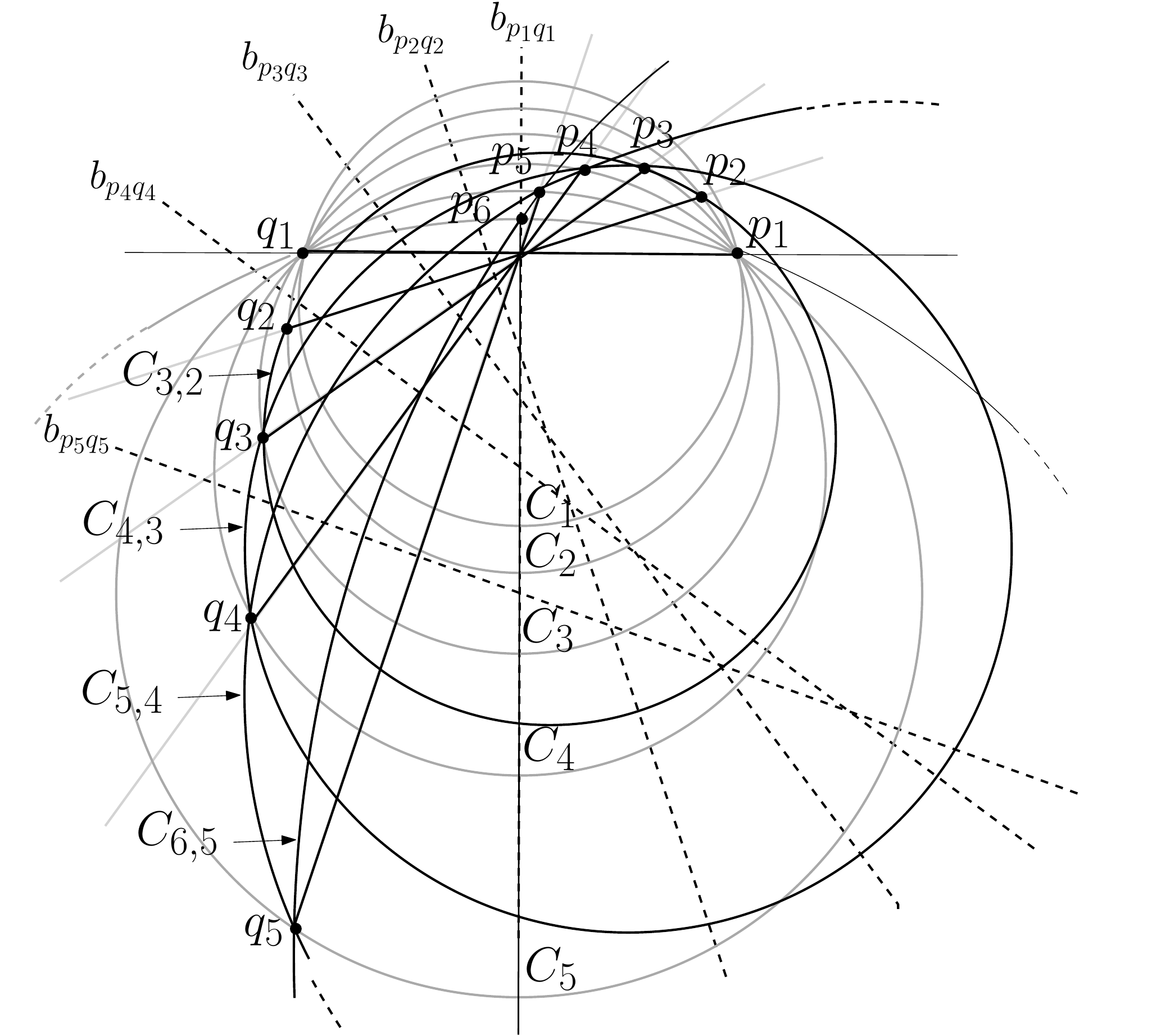}
				\caption{A set of $2n$ points. All $2n-1$ segments on the bisector of any $(n-1)$-edge have weight $n-2$,  $n-1$, or $n$.}
				\label{fig:halvings3}
			\end{center}
		\end{figure}
	
	Let  $2\leq k\leq n$. 
	From $C_{p_1q_1}$, we denote $C_k$ the circle through $p_1$, $q_1$ and the point $p_k$ on the bisector $b_{p_1q_1}$. Hence, $q_n$ is on $C_n$.  
	Note that these circles are well defined, once $p_k$ is defined. 
	
	Next, we place the point $q_{n-1}$ on line $\ell_{n-1}$, in the interior of $C_n$.
	We draw a circle through the points $q_n$, $p_n$ and $q_{n-1}$. This circle cuts $\ell_{n-1}$ in the point $q_{n-1}$ and in another one, where we put $p_{n-1}$.
	We denote $C_{i+1,i}$ the circle through  $q_{i+1},p_{i+1},q_i$ and $p_i$ for $i=n-1,\cdots,3,2$.  
	
	In a general step of this construction, we
	place the point $q_{n-i}$ on the line $\ell_{n-i}$, in the interior of $C_{n-i+1}$ and outside the circle $C_{n-i+2,n-i+1}$. Then we repeat the process to obtain the point $p_{n-i}$. See Figure~\ref{fig:halvings3}.
	
	Note that the points $q_{i+1},p_{i+1},q_i$ and $p_i$ are cocircular, they belong to $C_{i+1,i}$, by construction.
	Furthermore, $p_i$ and $q_i$ are the intersection points of $C_{i+1,i}$ with $C_{i,i-1}$. These points $p_i$ and $q_i$ define in  $C_{i,i-1}$ two arcs, the one containing $q_{i-1}$ is exterior to   $C_{i+1,i}$, the other one is in the interior of $C_{i+1,i}$.
	Therefore, if we go through any of the bisectors $b_{p_iq_i}$, the corresponding circle is reaching the pairs $p_i$ and $q_i$ in order. That is, if the circle passes through $q_{i-1}$ it can not reach $q_j$ (or $p_j$), where $j>i+1$, before $q_{i+1}$ (or $p_{i+1}$) is reached. 
	
	Note that we can move the points $q_i$ slightly to achieve general position, while keeping the described properties. 
	
	It is easy to check that the line through $\overline{p_iq_i}$ is a halving line and if we go through any of the bisectors $b_{p_iq_i}$, the corresponding circle  contains $n-2$, $n-1$ or $n$ points.
\end{proof}

%\begin{acknowledgements}
%\end{acknowledgements}




\end{document}